\documentclass{article}
\usepackage[utf8]{inputenc}

\usepackage[a4paper, total={6in, 9in}]{geometry}

\usepackage{amsmath}
\usepackage{amsfonts}
\usepackage{amssymb}
\usepackage{mathtools}
\usepackage{amsthm}
\usepackage{xparse,etoolbox}
\usepackage{bm}
\usepackage{xcolor}
\usepackage{bbm}
\usepackage{enumitem}
\usepackage{subcaption}
\usepackage{tikz}
\usepackage{pgfplots}
\usepgfplotslibrary{fillbetween}
\usetikzlibrary{patterns}
\pgfplotsset{compat=1.5.1}

\newcommand{\cube}[1]{\mathbb{F}_2^{#1}}
\newcommand{\R}{\mathbb{R}}
\newcommand{\N}{\mathbb{N}}
\newcommand{\1}{\mathbbm{1}}

\newcommand{\ve}{\varepsilon}

\newcommand{\Forb}[3]{\textup{ForbConfig}(#1,#2,#3)}
\newcommand{\ForbLin}[3]{\textup{ForbConfig}_{\textup{Lin}}(#1,#2,#3)}
\newcommand{\Valid}[3]{\textup{AllowedConfig}(#1,#2,#3)}
\newcommand{\ValidLin}[3]{\textup{AllowedConfig}_{\textup{Lin}}(#1,#2,#3)}
\newcommand{\Delsarte}[3]{\textup{Delsarte}(#1,#2,#3)}
\newcommand{\DelsarteLin}[3]{\textup{Delsarte}_{\textup{Lin}}(#1,#2,#3)}
\newcommand{\Lin}{\textup{Lin}}

\newtheorem{proposition}{Proposition}

\newtheorem{corollary}{Corollary}
\newtheorem{remark}{Remark}
\newtheorem{prob}{Problem}
\newtheorem{fact}{Fact}

\newcommand{\blocktheorem}[1]{%
  \csletcs{old#1}{#1}% Store \begin
  \csletcs{endold#1}{end#1}% Store \end
  \RenewDocumentEnvironment{#1}{o}
    {\par\addvspace{1.5ex}
     \noindent\begin{minipage}{\textwidth}
     \IfNoValueTF{##1}
       {\csuse{old#1}}
       {\csuse{old#1}[##1]}}
    {\csuse{endold#1}
     \end{minipage}
     \par\addvspace{1.5ex}}
}

\blocktheorem{proposition}
\blocktheorem{definition}
\blocktheorem{thm}
\blocktheorem{corollary}
\blocktheorem{remark}
\blocktheorem{prob}
\blocktheorem{fact}

\title{Linear Programming Hierarchies in Coding Theory:\\
Dual Solutions}
\author{
    Elyassaf Loyfer\thanks{School of Computer Science and Engineering, Hebrew University, 91904 Jerusalem, Israel. Supported in part by grant 659/18 “High-dimensional combinatorics” of the Israel Science Foundation.}
    \and
    Nati Linial\footnotemark[1]
}

\begin{document}

\maketitle

\begin{abstract}
The rate vs.\ distance problem is a long-standing open
problem in coding theory. Recent papers have suggested
a new way to tackle this problem by appealing to a new 
hierarchy of linear programs. If one can find good dual solutions
to these LP's, this would result in improved 
upper bounds for the rate vs.\ distance problem
of linear codes.
In this work we develop the first dual feasible 
solutions to the LP's in this hierarchy. These match
the best known bound for a wide range of parameters.
Our hope is that this is 
a first step towards better solutions, and
improved upper bounds for the rate vs.\ distance problem
of linear codes.
\end{abstract}

\section{Introduction}

% A binary of code $C$ of length $n$ is a subset of $\cube{n}$.
% The minimal distance of $C$ is the minimal Hamming distance 
% between any two distinct words in $C$.
% Let $A(n,d)$ be
% the maximal size of a code $C\subset\cube{n}$ with minimal distance
% $d$. 
% If $$

% A fundamental problem in coding theory is to find the
% asymptotic maximal rate of a code with relative distance $\delta$,
% namely
% \[
%     \mathcal{R}(\delta) 
%     \coloneqq
%     \limsup_{n\to\infty}
%     \frac{1}{n}
%     \log_2 \left( A(n,\lfloor \delta n \rfloor) \right)
% \]

The {\em rate vs.\ distance problem} is a major open problem in coding theory.
It seeks the largest cardinality $A(n,d)$ 
of a binary code of length $n$ with minimal distance $d$.
Here we are interested in the range $d=\Theta(n)$ and $n\to\infty$.
In this case, $A(n,d)$ is known to
grow exponentially in $n$, and
%the relevant quantity is
we consider 
the {\em asymptotic maximal rate},
% with relative distance $\delta$,
% It is known that $A(n,d)$ grows exponentially
% The asymptotic version of the problem, known as the rate vs. distance
% problem, is to find the best possible rate of such a code, when
% $n\to\infty$ and $d$ is a fixed fraction of $n$,
% namely to find
\[
    \mathcal{R}(\delta) 
    \coloneqq
    \limsup_{n\to\infty}
    \frac{1}{n}
    \log_2 \left( A(n,\lfloor \delta n \rfloor) \right)
\]
where $0<\delta<1/2$ is the {\em relative distance} of the code.

% Finding $\mathcal{R}(\delta)$ is known as the rate-vs.-distance
% problem.
The best known lower bound $\mathcal{R}(\delta)\ge 1-H(\delta)$
% , which is proved by showing existence, 
was given
by Gilbert \cite{gilbert1952comparison} for general codes and
by Varshamov \cite{varshamov1957estimate} for linear codes,
where $H$ is the binary entropy function.

% The best upper bound is based on a linear program (LP) by Delsarte
% \cite{delsarte1973algebraic}. The optimum of Delsarte's LP
% upper bounds $A(n,d)$.
% Codes can be defined in different metric spaces, e.g.\ the Hamming
% cube $\cube{n}$, slices of the cube, or the unit sphere $S_n$.
% To each metric space there corresponds a different version of
% Delsarte's LP. In this work we focus on binary codes.

The best known upper bounds are {\em the first and second linear programming
(LP) bounds} \cite{mceliece1977new}, both of which are based on 
Delsarte's linear program \cite{delsarte1973algebraic}.
% They were obtained in \cite{mceliece1977new}
% based on Delsarate's linear programming approach \cite{delsarte1973algebraic}.
The first LP bound $$\mathcal{R}(\delta)\leq H(1/2-\sqrt{\delta(1-\delta)})$$ is the best known upper bound
for $0.273<\delta<1/2$. Much of what we do here revolves around this bound. 
% Its proof 
% obtained by constructing 
% is by construction
% of a dual
% feasible solution to Delsarte's LP.
The exact value of Delsarte's LP remains unknown. However,
there is strong numerical evidence \cite{barg1999numerical}
that the MRRW \cite{mceliece1977new} bound has fully exhausted its potential to upper bound $\mathcal{R}(\delta)$.

A code $C\subset\cube{n}$ is {\em linear} if it is a linear subspace.
This is, of course a very strong restriction, so it stands to reason
that one should be able to derive stricter upper bounds that are specific to 
linear codes.
We denote by $A_\Lin(n,d)$ and $\mathcal{R}_{\Lin}(\delta)$
the analogues of $A(n,d)$ and $\mathcal{R}(\delta)$ when restricted to
linear codes.
Recent works \cite{coregliano2021complete,loyfer2022new} are opening
the way to linear programs stronger than Delsarte's that hopefully improve the
upper bound {\em for linear codes}.
% the largest possible size of a linear
% code of length $n$ and minimal distance $d$.

Coregliano et.\ al.\
\cite{coregliano2021complete} developed a new hierarchy of linear 
programs whose $\ell$-th member
upper-bounds $A_{\Lin}(n,d)^{\ell}$,
% ($A_{\Lin}(n,d)$ to the $\ell$-th power).
and converges to this quantity when $\ell = \Omega(n^2)$.
% Moreover, they proved that
% the hierarchy converges to $A_{\Lin}(n,d)^{\ell}$ 
% when $\ell = \Omega(n^2)$. 
The novel idea behind the new hierarchy is to consider the
Cartesian product of $\ell$ copies of a code. 
This way, the linearity property of the code can be
utilized in addition to Delsarte's
constraints.
% If the linearity assumption is lifted, the resulting hierarchy
% applies for general codes, but is degenerate, namely it is not
% better than Delsarte's LP.

% A parallel hierarchy for general codes was also presented in which
% all levels are equivalent in value to Delsarte's LP.

% The idea behind the new LPs is to consider cartesian products of the code.
% Thus, the linearity property can be utilized in addition to Delsarte's
% constraints, which still apply.
% If the code is not assumed to be linear, a parallel hierarchy
% can be formed in which all levels are equivalent to Delsarte's LP.

In our previous work \cite{loyfer2022new}
we employ related ideas to develop a
% but a somewhat different approach,
hierarchy which is stricter than that of \cite{coregliano2021complete}.
% with a richer set of constraints. 
We also suggest
% They consider several modifications to the
% hierarchy, one of which is 
a different objective function
that bounds $A_{\Lin}(n,d)$ instead of $A_{\Lin}(n,d)^{\ell}$.
In the present paper we consider both objective functions.
% We find this objective function interesting, for two reasons. First,
% it is easier to find a dual feasible solution. Second, it
% is well defined when $\ell$ tends to infinity.

% Both objectives are considered here, and in section \textcolor{blue}{TODO:
% ref discussion} we discuss the pros and cons of each.
% \textcolor{red}{ [Maybe not here?]
% The completeness proof of \cite{coregliano2021complete}
% does not apply to the modified objective. However, numerical
% experiments show that the two objectives are comparable
% when $\ell=2$ \cite{loyfer2022new}.
% }

The LP hierarchies are extremely interesting as they may 
lead to progress in the longstanding problem of
bounding $\mathcal{R}_{\Lin}(\delta)$.
One natural course of action in this direction
% Ultimately, we hope that
% the LP hierarchies will advance our
% knowledge of $\mathcal{R}_{\Lin}(\delta)$. 
% A natural course
% of action 
is to find good dual feasible solutions,
the same way that the LP bounds were proven based on Delarte's LP.
% as did MRRW with Delsarte's LP.
It is challenging to find dual feasible solutions for the LP hierarchies.
Due to higher dimensionality, and increased complexity of the LP's.
It is even far from trivial 
to find dual feasible solutions
for the hierarchy which attains the first LP bound.
In this work we make a first step in this direction.

\subsection{Our Contribution}

\begin{enumerate}
    \item\label{first_point}
    We construct a family of dual feasible solutions
    for the LP hierarchy, which attain the first LP bound
    up to $\ell \leq \log n- \log\log n$, where $\ell$
    is the level in the hierarchy.
    These solutions 
    apply for both linear and non-linear codes.
    
    It is natural to ask how to apply this method to linear codes,
    and we provide a partial answer to this question.

    \item We consider the alternative objective function,
    which bounds $A_{\Lin}(n,d)$ instead of $A_{\Lin}(n,d)^\ell$,
    and construct a family of feasible solutions.
    
    In contrast with the solution alluded to in point \ref{first_point}
    these solutions apply to {\em all} values of $\ell$. Also, while both
    approaches rely on solutions to Delsarte's LP, this one treats these
    solutions as black boxes.

\end{enumerate}

\subsection{Outline of the Paper}

In section \ref{section:preliminaries} we provide preliminary
material, including the relevent LP hierarchies (\ref{section:hierarchies_overview}), and
a dual feasible solution to Delsarte's LP which
establishes the first LP bound (\ref{section:first_lp_bound}).
In section \ref{section:dual_sol_ell} we construct a 
dual feasible solution to the LP hierarchy for general codes,
and discuss how similar techniques can be applied
to linear codes.
In section \ref{section:dual_sol_lin} we provide a
dual feasible solution to the LP hierarchy
for linear codes, with an objective function that
is linear in $A_{\Lin}(n,d)$. We close with
some concluding remarks in section \ref{section:discussion}.

Proofs are deferred to the end of the paper, in appendix
\ref{section:proofs}.

% % Since a good solution for the $\ell=1$ step exists, we use it construct
% % equivalent solutions for larger values of $\ell$, in particular:
% \begin{itemize}
%     \item We construct a family of dual feasible solutions
%     for the LP hierarchy which are equivalent to MRRW,
%     up to $\ell \leq \log n- \log\log n$, where $\ell$
%     is the level in the hierarchy.
    
%     \item The solutions
    
%     \item We then employ this framework to construct solutions
%     for the LP hierarchy which bounds $A(n,d)^\ell$.
%   The proposed
%     solutions are asymptotically equivalent to the MRRW bound
%     up to $\ell \leq \log n - \log\log n$.
%     This construction applies to both linear and non-linear codes.

%     \item We develop a reduction that is specific to linear codes,
%     and the hierarchy which bounds $A_{\Lin}(n,d)^\ell$.

%     \item We construct a dual feasible solution to the LP hierarchy
%     which bounds $A_{\Lin}(n,d)$, for every $\ell \geq 1$. Starting from
%     any solution for $\ell=1$, this construction
%     yields an equivalent solution for 
%     any $\ell$. In particular, the MRRW solution for $\ell=1$, yields
%     this way an equivalent solution for any $\ell$. 
% \end{itemize}

\section{Preliminaries}
\label{section:preliminaries}

A binary code of length $n$ is a subset $C\subset\cube{n}$. Throughout
we only discuss binary codes.
We denote by $|x|\coloneqq |\{1\leq i \leq n: x_i\neq 0\}|$ the Hamming weight 
of $x\in\cube{n}$. The Hamming distance between $x\in\cube{n}$ and $y\in\cube{n}$
is $|x+y|$. The code's distance is $dist(C)\coloneqq \min \{|x+y|: x,y\in C,~x\neq y\}$. 
The largest possible size of a binary code of length $n$ and
distance $d$ is denoted 
\[
A(n,d) \coloneqq \max
\left\{
|C| : C\subset\cube{n},~ dist(C)\geq d
\right\}.
\]
A code is {\em linear} if it is a linear subspace.
For linear codes, this size is denoted $A_{\Lin}(n,d)$. 

The rate of a code is $R(C)\coloneqq \frac{1}{n}\log_2|C|$. The
rate vs. distance problem is to find
\[
    \mathcal{R}(\delta) 
    \coloneqq
    \limsup_{n\to\infty}
    \frac{1}{n}
    \log_2 \left( A(n,\lfloor \delta n \rfloor) \right)
\]
for every $\delta\in(0,1/2)$.

Let $f,g:\cube{n}\to\R$. We define inner product w.r.t.\
the uniform measure, $\langle f,g\rangle \coloneqq 2^{-n}\sum_{x\in\cube{n}} f(x)g(x)$. The convolution between $f,g$ is
denoted $f*g$ and defined by
$(f*g)(x) \coloneqq 2^{-n}\sum_y f(y)g(x+y)$.

The Fourier transform 
of $f$ is denoted either $\mathcal{F}(f)$ or $\widehat{f}$ and defined by
$\widehat{f}(x) \coloneqq \langle f,\chi_x\rangle$, where $\chi_x(y) = (-1)^{\langle x, y \rangle}$. Fourier transform is its own inverse,
up to normalization: $2^n \mathcal{F}(\mathcal{F}(f))= f$.
In Fourier domain, inner product and convolution are without normalization, namely
$\langle \widehat{f},\widehat{g}\rangle_{\mathcal{F}} =
\sum_{x}\widehat{f}(x)\widehat{g}(x)$
and $(\widehat{f}*_{\mathcal{F}} \widehat{g}) (x)
= \sum_{y}\widehat{f}(y) \widehat{g}(x+y)$. In favor of
readability we omit the subscript $\mathcal{F}$ when possible.

By the Convolution Theorem, 
$\widehat{f*g} = \widehat{f}\cdot \widehat{g}$.
Similarly, $\widehat{f\cdot g}=
\widehat{f}*_{\mathcal{F}}\widehat{g}$.

Let $\ell \in \N$. We identify the space $\cube{\ell n}$ with
the spaces $\cube{\ell \times n}$ and $\left(\cube{n}\right)^{\ell}$.
Given $X \in \cube{\ell \times n}$ we denote its rows by
$x_1,\dots,x_\ell$. Given a function $f:\cube{\ell n}\to \R$,
we sometimes write $f(X)$ and other times $f(x_1,\dots,x_\ell)$,
both have the same meaning.

\subsection{Krawtchouk Polynomials}
\label{section:karwtchouk}

The (univariate) Krawtchouk polynomials
$\{K_0\equiv 1,K_1,\dots,K_n\}$
are a family of
orthogonal polynomials w.r.t.\ binomial measure,
\[
    \sum_{i=0}^{n} 2^{-n}\binom{n}{i} K_j(i) K_k(i)
    = \binom{n}{j} \delta_{j,k}
\]
The Krawtchouks are uniquely determined up to normalization.
Here we assume the normalization $K_{i}(0) = \binom{n}{i}$,
for $i=0,\dots,n$.
The Krawtchouks are defined over $\R$ but we extend their definition
to the cube, writing $K_i(x) \coloneqq K_i(|x|)$,
for $x\in\cube{n}$.

The Fourier transform of the $i$-th
Krawtchouk polynomial is the $i$-th {\em level-set indicator} $L_i$,
\[
    L_i(x) \coloneqq \1_{|x|=i},\quad
    \widehat{K}_i(x) = L_i(x)
\]

\subsection{Overview of LP Hierarchies}

\label{section:hierarchies_overview}

We describe the LP hierarchies related to the current work, without
proofs. More details can be found
in \cite{coregliano2021complete,loyfer2022new}.

All of the hierarchies are parameterized by three positive integers:
$n$ - the code's length; 
$d$ - the code's distance; and
$\ell$ - the level in the hierarchy.
Every LP in the hierarchy can be symmetrized and converted to an
equivalent LP with multivariate Krawtchouk polynomials. For convenience
we use Fourier-theoretic terminology.

% Throughout let $n,d,\ell$ be positive integers and $d\leq n/2$.

\subsubsection{LP Hierarchy for General Codes}
Define the set of {\em forbidden configurations} as in
\cite{coregliano2021complete},
\[
    \Forb{n}{d}{\ell} 
    \coloneqq
    \left\{
    (x_1,\dots,x_\ell)\in \left(\cube{n}\right)^{\ell}
    :
    1\leq|x_i|\leq d-1 ~\text{for some}~ i=1,\dots,\ell
    \right\}
\]

Denote the following LP by $\textup{Delsarte}(n,d,\ell)$. Its
optimum is an upper bound on $A(n,d)^{\ell}$.
\begin{align*}
    & \textup{maximize} && \sum_{X\in\cube{\ell\times n}} f(X)
    \\
    & \textup{subject to} && f:\cube{\ell\times n} \to\R
    \\
    &&& f(0) = 1
    \\
    &&& f \geq 0 
    \\
    &&& \widehat{f} \geq 0
    \\
    &&& f(X) = 0 & \text{if}~ X\in \Forb{n}{d}{\ell}
\end{align*}
We note that this hierarchy is degenerate, namely
$\textup{Delsarte}(n,d,\ell) = \textup{Delsarte}(n,d,\ell+1)$
for every $\ell$. The equality is between the optimal values.

\subsubsection{LP Hierarchies for Linear Codes}
Define the set of forbidden configurations for linear codes,
\[
    \ForbLin{n}{d}{\ell} 
    \coloneqq
    \left\{
    X \in \cube{\ell \times n} :
    1\leq|x|\leq d-1
    ~\text{for some}~ x \in rowspan(X)
    \right\}
\]

Denote the following LP by $\textup{Delsarte}_{\textup{Lin}}(n,d,\ell)$.
Its optimum is an upper bound on $A_{Lin}(n,d)^{\ell}$.
\begin{align*}
    & \textup{maximize} && \sum_{X\in\cube{\ell\times n}} f(X)
    \\
    & \textup{subject to} && f:\cube{\ell\times n} \to\R
    \\
    &&& f(0) = 1
    \\
    &&& f \geq 0 
    \\
    &&& \widehat{f} \geq 0
    \\
    &&& f(X) = 0 & \text{if}~ X\in \ForbLin{n}{d}{\ell}
\end{align*}
Note that 
$\textup{Delsarte}(n,d,\ell)$ and
$\textup{Delsarte}_{\textup{Lin}}(n,d,\ell)$ differ only in their sets of
forbidden configurations, (which is larger in the linear case).

The final LP hierarchy that we consider has the same
set of constraints as
$\textup{Delsarte}_{\textup{Lin}}(n,d,\ell)$, 
but a different objective function. Its optimum is
an upper bound to $A_{Lin}(n,d)$, rather than
$A_{Lin}(n,d)^{\ell}$. Namely, it is related linearly
to the code's size.
\begin{align*}
    & \textup{maximize} && \frac{1}{2^\ell-1}
    \sum_{0\neq u\in\cube{\ell}}
    \sum_{x\in\cube{n}} f(ux^\top)
    \\
    & \textup{subject to} && f\in \textup{Delsarte}_{\textup{Lin}}(n,d,\ell)
\end{align*}
Here $ f\in \textup{Delsarte}_{\textup{Lin}}(n,d,\ell)$
means that $f$ is feasible for
$\textup{Delsarte}_{\textup{Lin}}(n,d,\ell)$.
\subsection{The First LP Bound}

\label{section:first_lp_bound}

The first LP bound is obtained by constructing a
dual feasible solution to Delsarte's LP.
In this section we present such a construction, which will
be used in the subsequent section.
% We will also consider a more abstract view of the solution and
% how it can be applied to related LPs, i.e.\
% $\DelsarteLin{n}{d}{\ell}$.

The dual of Delsarte's LP, for binary codes of length $n$ and distance
$d$, can be presented as follows.
\begin{proposition}
    \label{prop:delsarte_dual}
    $A(n,d)$ is upper bounded by
    \begin{align}
        % A(n,d) \leq~ 
        & \textup{minimize} && g(0)/\widehat{g}(0)
        \label{eq:delsarte_dual1}
        \\
        & \textup{subject to} && g:\cube{n}\to\R
        \label{eq:delsarte_dual2}
        \\
        &&& \widehat{g} \geq 0
        \label{eq:delsarte_dual3}
        \\
        &&& \widehat{g}(0) > 0
        \label{eq:delsarte_dual4}
        \\
        &&& g(x) \leq 0 && \text{if}~ |x|\geq d
        \label{eq:delsarte_dual5}
    \end{align}
\end{proposition}
To turn it into an LP, we can further posit that $\widehat{g}(0)=1$.

All of the solutions to this dual LP,
given in 
\cite{mceliece1977new,navon2005delsarte,navon2009linear,samorodnitsky2021one,barg2006spectral,barg2008functional},
have the form
\begin{equation}
    \label{eq:g_universal_form}
    g(x) = (t-|x|)\cdot\Lambda^2(x)
\end{equation}
where $t\leq d$, and $\Lambda$ is chosen appropriately. This guarantees 
that constraint \eqref{eq:delsarte_dual5} is satisfied,
and it only remains to find $\Lambda$ that satisfies the
Fourier constraints, \eqref{eq:delsarte_dual3} and \eqref{eq:delsarte_dual4}. The linearity of the function 
$x\mapsto (t-|x|)$ simplifies this task.

The above-mentioned solutions also share the same $\Lambda$,
with slight differences. But the different methods used to construct
this $\Lambda$ shed new light over the approach
given in \eqref{eq:g_universal_form} 
which originated in
\cite{mceliece1977new}. 
As we explain shortly,
% it is no coincidence that
the function 
% $(t-|x|)$ 
$x \mapsto |x|$
is related to the adjacency matrix
of the Hamming cube.  
Also, a good choice for $\Lambda$ is the first eigenfunction
of the smallest Hamming ball
which satisfies a certain constraint.

To see this connection, note that $2(t-|x|) = K_1(x)-K_1(t)$, where
$K_1(t)=n-2t$ is the first Krawtchouk polynomial. The Fourier
transform of $K_1$ is $L_1$, the indicator function of the set
$\{x\in\cube{n}: |x|=1\}$. Consider the operator of convolution with $L_1$.
The matrix of this operator 
is the  $2^n\times 2^n$ matrix $A$, the adjacency matrix of Hamming cube.
Namely, for any $x, y\in \cube{n}$
\begin{equation}
    \label{eq:A_def}
    A_{x,y} =
    \begin{cases}
        1 & |x+ y| = 1 \\
        0 & \text{otherwise}
    \end{cases}
\end{equation}
We include the simple proof: let $f:\cube{n}\to\R$,
\begin{equation}
    \label{eq:proof_K1_fourier}
    2^n(L_1*f)(x) = \sum_{y\in\cube{n}} L_1(y) f(x+y)
    = \sum_{i=1}^{n} f(x+e_i)
    = \sum_{y: |y+x|=1} f(y)
    = (Af)(x)
\end{equation}
All papers \cite{mceliece1977new,navon2005delsarte,navon2009linear,samorodnitsky2021one,barg2006spectral,barg2008functional}
% need to  find 
% conditions on $\Lambda$ under which $g=(t-|x|)\Lambda^2$ is feasible, and then 
find an
appropriate $\Lambda$ to establish the first LP bound. Of all these papers
our approach is closest to that of \cite{navon2005delsarte}.

\begin{proposition}
    \label{prop:lambda_sufficient_conditions}
    Let $\ve > 0$.
    Let $\Lambda:\cube{n}\to\R$ such that
    \[
    \textup{(a)}~ \widehat{\Lambda}(0) = 1;
    \qquad
    \textup{(b)}~ \widehat{\Lambda} \geq 0;
    \qquad
    \textup{(c)}~ A\cdot \widehat{\Lambda} \geq (n-2d+2\ve) \widehat{\Lambda};
    \]
    Then, $g(x)\coloneqq 2(d-|x|)\Lambda^2(x)$ is a feasible
    solution to Delsarte's dual LP, and 
    \[
        \frac{g(0)}{\widehat{g}(0)} \leq \frac{d}{\ve} \left|supp(\widehat{\Lambda})\right|
    \]
\end{proposition}
\begin{proposition}
    \label{prop:lambda_existence}
    There exists a function $\Lambda = \Lambda_{d,\ve}$
    which satisfies proposition
    \ref{prop:lambda_sufficient_conditions}, and its
    Fourier transform,
    $\widehat{\Lambda}$,
    is supported
    on the Hamming ball of radius 
    $r = n/2-\sqrt{d(n-d)})+o(n)$.
    % where $\delta = d/n$.
\end{proposition}
\begin{corollary}[The First LP Bound]
    \label{cor:first_lp_bound}
    \[
        \mathcal{R}(\delta)\leq H(1/2+\sqrt{\delta(1-\delta)})
    \]
\end{corollary}
Let us describe a function $\Lambda$ for 
proposition \ref{prop:lambda_existence}.
Let $A^{\leq r}$ be a submatrix of $A$ corresponding to all 
vertices $x\in\cube{n}$ of Hamming weight $\le r$.
Namely, the adjacency matrix of the Hamming
ball of radius $r$.
We choose $\Lambda$ such that $\widehat{\Lambda}$
% The function $\Lambda$ which satisfies proposition
% \ref{prop:lambda_existence}
is the Perron eigenfunction of $A^{\leq r}$,
and pick the smallest $r$ for which $A^{\leq r}$
has spectral radius at least $n-2(d-\ve)$. For more details,
see the proof of proposition \ref{prop:lambda_existence} and
the remark that follows, in appendix \ref{section:proofs}.

\section{Dual Solutions to the LP Hierarchies}
% $\Delsarte{n}{d}{\ell}$}

\label{section:dual_sol_ell}

In this section we construct a family of dual feasible solutions
for the LP hierarchy $\Delsarte{n}{d}{\ell}$. We also 
consider how to apply the same ideas to 
% discuss a way to 
$\DelsarteLin{n}{d}{\ell}$,
and the resulting complications.

As in Delsarte's dual LP, also the duals of the hierarchies,
which we define below,
consist of two types of constraints: Fourier constraints,
and a non-positivity constraint.
Thus, we may again try to decompose $g$ into a function which guarantees
non-positivity $(d-|x|)$, and a function geared at
yielding the Fourier constraints.
However, while for $\ell=1$ a linear function is all you need for the non-positivity constraint,
this is no longer possible when $\ell$ grows.

Instead of a linear function, we construct a polynomial
$\Phi_{n,d,\ell}$ which is non-positive in the desired regions,
and seek a function $g_{n,d,\ell}:\cube{\ell\times n}\to\R$
of the form
\[
    g_{n,d,\ell} = \Phi_{n,d,\ell} \cdot \Gamma_{n,d,\ell}^2
\]
as a dual feasible solution to $\Delsarte{n}{d}{\ell}$. It turns
out that for our choice of $\Phi$, the function
$\Gamma \coloneqq \Lambda^{\otimes \ell} = \Lambda \otimes 
\dots \otimes \Lambda$ works, where $\Lambda$ 
is from proposition \ref{prop:lambda_existence}.
In other words, the solution is obtained by a reduction
from the $\ell$-th level to Delsarte.

The main shortcoming of our solution is its fast growth in $n$:
\[
    value(g_{n,d,\ell})^{1/\ell}
    \leq 
    % \left(e n^{1/\delta}\right)^{\frac{2^\ell\log\ell }{\ell}}
    % \left|supp(\widehat{\Lambda})\right|
    % = 
    \left(e n^{1/\delta}\right)^{\frac{2^\ell\log\ell }{\ell}}
    2^{nH\left(1/2 + \sqrt{\delta(1-\delta)}\right) + o(n)}
\]
where $\delta = d/n$.
When $\ell$ is too large,
the first term becomes dominant and the solution's value exponentially exceeds
the first LP bound.
Moreover, the hierarchy for general codes is known to be
degenerate, namely, comparing optimal values,
\[
    \left(\Delsarte{n}{d}{\ell}\right)^{1/\ell}
    =\Delsarte{n}{d}{1}
\]
which means, in particular,
that there exists a solution to the $\ell$-th level
that has the exact same value of the solution
from the previous section.

So, do the methods that we use
for general codes apply to linear codes as well?
In this case, we are able to construct an analogue of 
$\Phi$ that is suitable for linear codes, but
a solution based on $\Lambda^{\otimes \ell}$ no 
longer works. Instead, we suggest a reduction to a
problem of the same spirit of
proposition \ref{prop:lambda_sufficient_conditions}.

Throughout this section, we fix the parameters
$n,d,\ell$, and omit their subscripts, e.g.\
we write $g$ instead of $g_{n,d,\ell}$.
Also, we denote $\delta = d/n$.

\subsection{General Codes - $\Delsarte{n}{d}{\ell}$}

Let us first define the dual of $\Delsarte{n}{d}{\ell}$.
\begin{proposition}
    \label{prop:dual_ell_general}
    $A(n,d)^\ell$ is upper bounded by
    \begin{align}
        & \textup{minimize} && g(0)/\widehat{g}(0)
        \label{eq:dual_ell1}
        \\
        & \textup{subject to} && g:\cube{\ell \times n} \to \R
        \label{eq:dual_ell2}
        \\
        &&& \widehat{g} \geq 0 
        \label{eq:dual_ell3}
        \\
        &&& \widehat{g}(0) > 0
        \label{eq:dual_ell4}
        \\
        &&& g(X) \leq 0 && X\in \Valid{n}{d}{\ell}\wedge X\neq 0 % \setminus\{0\}
        \label{eq:dual_ell5}
    \end{align}
    where $\Valid{n}{d}{\ell}$ is the complement of the set
    of forbidden configurations,
    \[
        \Valid{n}{d}{\ell}
        \coloneqq 
        \left\{
            (x_1,\dots,x_\ell)\in\left(\cube{n}\right)^\ell:
            x_i = 0 \vee |x_i|\geq d
            \text{ for all }i=1,\dots,\ell
        \right\}
    \]
\end{proposition}
We proceed to construct a feasible solution in two steps:
% Following the ideas
% % We use the ideas 
% from section \ref{section:first_lp_bound},
% the search of an analytic solution consists of 
% two steps:
\begin{enumerate}[label=(\Roman*)]
    \item\label{step1} Define a function $\Phi:\cube{\ell\times n}\to\R$ such that $\Phi(0)>0$ and $\Phi$ satisfies constraint \eqref{eq:dual_ell5}.
    \item\label{step2} Find a function $\Gamma:\cube{\ell\times n}\to\R$ such that
    \[
        g\coloneqq \Phi\cdot \Gamma^2
    \]
    is feasible, namely it satisfies
    constraints \eqref{eq:dual_ell3} and \eqref{eq:dual_ell4} (the remaining
    constraint is satisfied by construction).
\end{enumerate}
Of course we want to carry out step \ref{step1} with a function $\Phi$,
that makes step \ref{step2} possible. 
% For Delsarte's LP
% ($\ell=1$) the linear
% function $\phi(x)=2(d-|x|)$ works. 

% We fix the parameters $n,d,\ell$,
% and denote $\delta = d/n$ and turn to 

% Let us carry out steps \ref{step1} and \ref{step2}.

Here is the idea behind our construction of $\Phi$ (see fig. \ref{fig:phi}):
Consider a set of $2^\ell-1$ balls, 
each in one subcube $\{\cube{U\times n}\}_{U\subset[\ell]}$
of $\cube{\ell\times n}$. Pick the centers, the radii and the $\ell_p$-norm
of the balls so that if $X\in\Valid{n}{d}{\ell}$, it is contained
in an {\em odd} number of balls.
For each ball define
a function which is negative inside the ball, and positive
outside of it. Finally, $\Phi$ is the product of these functions.
If $X\in\Valid{n}{d}{\ell}$, then $\Phi(X)\leq 0$, since it is the product
of and odd number of non-positive functions, and an even
number of non-negative functions.

\begin{figure}
\centering
    % \hfill
    \begin{subfigure}[b]{0.45\textwidth}
        \begin{tikzpicture}
        
        \begin{axis}[ 
            title={$\Valid{n}{d}{2}$},
            ticks=none,
            axis lines = middle,
            ymin=-0.5, ymax=6,
            xmin=-0.5, xmax=6,
            xlabel={$|x_1|$},
            ylabel={$|x_2|$},
            % axis x line*=bottom,
            % axis y line*=left,
            axis equal image
        ]
        % \addplot[black, domain=-2:2] {1};
        % \addplot[black, domain=-2:2] {y=1};
        \node[anchor=north west] at (axis cs:0.9,0) {$d$};
        \node[anchor=north west] at (axis cs:4.9,0) {$n$};
        \node[anchor=south east] at (axis cs:0,0.9) {$d$};
        \node[anchor=south east] at (axis cs:0,4.9) {$n$};
        \addplot[dashed, samples=2 ,gray, name path=three] coordinates {(1,-1)(1,6)};
        \addplot[dashed, samples=2 ,gray, name path=three] coordinates {(-1,1)(6,1)};
        \addplot[dashed, samples=2 ,gray, name path=three] coordinates {(5,-1)(5,6)};
        \addplot[dashed, samples=2 ,gray, name path=three] coordinates {(-1,5)(6,5)};
        \fill[red, opacity=0.3] (axis cs:1,1) rectangle (axis cs:5,5);
        \addplot[line width=6, opacity=0.3, samples=2 ,red, name path=three] coordinates {(0,1)(0,5)};
        \addplot[line width=6, opacity=0.3, samples=2 ,red, name path=three] coordinates {(1,0)(5,0)};
        \end{axis}
        \end{tikzpicture}
    \end{subfigure}
    \hfill
    \begin{subfigure}[b]{0.45\textwidth}
        \begin{tikzpicture}
        
        \tikzstyle{reverseclip}=[insert path={(current page.north east) --
          (current page.south east) --
          (current page.south west) --
          (current page.north west) --
          (current page.north east)}
        ]
        
        \begin{axis}[ 
            title={$\Phi_{n,d,2}\leq 0$},
            ticks=none,
            axis lines = middle,
            ymin=-0.5, ymax=6,
            xmin=-0.5, xmax=6,
            xlabel={$|x_1|$},
            ylabel={$|x_2|$},
            axis equal image
        ]
        \node[anchor=north east] at (axis cs:1.1,0) {$d$};
        \node[anchor=north west] at (axis cs:4.9,0) {$n$};
        \node[anchor=north east] at (axis cs:0,1.1) {$d$};
        \node[anchor=south east] at (axis cs:0,4.9) {$n$};
        \addplot[dashed, samples=2 ,gray, name path=three] coordinates {(1,-1)(1,6)};
        \addplot[dashed, samples=2 ,gray, name path=three] coordinates {(-1,1)(6,1)};
        \addplot[dashed, samples=2 ,gray, name path=three] coordinates {(5,-1)(5,6)};
        \addplot[dashed, samples=2 ,gray, name path=three] coordinates {(-1,5)(6,5)};
        \fill[red, opacity=0.3] (axis cs:1,1) rectangle (axis cs:5,5);
        \begin{scope}
            \clip (axis cs:3,3) circle (2.828) [reverseclip];
            \fill[red, opacity=0.3] (axis cs:-1,1) rectangle (axis cs:1,5);
        \end{scope}
        \begin{scope}
            \clip (axis cs:3,3) circle (2.828) [reverseclip];
            \fill[red, opacity=0.3] (axis cs:5,1) rectangle (axis cs:6,5);
        \end{scope}
        \begin{scope}
            \clip (axis cs:3,3) circle (2.828) [reverseclip];
            \fill[red, opacity=0.3] (axis cs:1,-1) rectangle (axis cs:5,1);
        \end{scope}
        \begin{scope}
            \clip (axis cs:3,3) circle (2.828) [reverseclip];
            \fill[red, opacity=0.3] (axis cs:1,5) rectangle (axis cs:5,6);
        \end{scope}
        \end{axis}
        \end{tikzpicture}
    \end{subfigure}
    
\caption{Illustration of
$\Valid{n}{d}{\ell}$ and
$\Phi_{n,d,\ell} \leq 0$, for $\ell = 2$.
The axes are the Hamming weights of $(x_1,x_2)\in\left(\cube{n}\right)^2$.
}
\label{fig:phi}
\end{figure}

Let us carry out steps \ref{step1} and \ref{step2}.

\noindent\textbf{Step \ref{step1}}.
Let $m\in\N$ be even such that $\ell \leq \left(\frac{1+\delta}{1-\delta}\right)^m$.
For every $\emptyset\neq U\subset [\ell]$, let
% Define $\phi_U:\cube{\ell \times n}\to\R$ as follows:
\[
    \phi_U(x_1,\dots,x_\ell) = \sum_{i\in U}
    \left[ (n+d-2|x_i|)^m-(n-d)^m \right]
\]

Define $\Phi=\Phi_{n,d,\ell}:\cube{\ell\times n}\to\R$:
\[
    \Phi= \prod_{\emptyset\neq U\subset [\ell]}\phi_U
\]

\noindent\textbf{Step \ref{step2}}.
Let $\ve >0$ such that $(n-d+2\ve)^m - (n-d)^m = 1$, i.e.
\[
    2\ve = ((n-d)^m+1)^{1/m} - (n-d)
\]
Let $\Lambda = \Lambda_{d,\ve}$ from proposition
\ref{prop:lambda_existence}.
Define $\Gamma \coloneqq \Lambda^{\otimes \ell}$,
the tensor product of $\ell$ copies of $\Lambda$.

The following propositions establish the main result of this section,
corollary \ref{cor:dual_ell}.

\begin{proposition}\label{prop:nonpos_poly_general}
    \hfill
    \begin{enumerate}
        \item \label{nonpos_poly_general1} $\phi_U(x_1,\dots,x_\ell) \geq 0$ if $x_i=0$ for some $i\in U$.
        \item \label{nonpos_poly_general2} $\phi_U(x_1,\dots,x_\ell) \leq 0$ if $|x_i|\geq d$ for all $i\in U$.
        \item \label{nonpos_poly_general3} $\Phi(0) > 0$.
        \item \label{nonpos_poly_general4} $\Phi(X) \leq 0$ if $X\in \Valid{n}{d}{\ell}$ and $X\neq 0$.
    \end{enumerate}
\end{proposition}

\begin{proposition}
    \label{prop:dual_feasibl_sol_ell}
    Let $g_{n,d,\ell} \coloneqq \Phi_{n,d,\ell} \cdot (\Lambda^{\otimes \ell})^2$.
    \begin{enumerate}
        \item $g_{n,d,\ell}$ is a feasible dual solution to
        $\Delsarte{n}{d}{\ell}$.
        \item The value of $g_{n,d,\ell}$ is
        \[
            value(g_{n,d,\ell}) 
            = \frac{g_{n,d,\ell}(0)}{\widehat{g}_{n,d,\ell}(0)}
            \leq \left(e n^{1/\delta}\right)^{2^{\ell} \log \ell}
            \left| supp(\widehat{\Lambda}) \right|^{\ell}
        \]
    \end{enumerate}
\end{proposition}
\begin{corollary}
    \label{cor:dual_ell}
    $value(g_{n,d,\ell})^{1/\ell}$ coincides with the first LP
    bound for $\ell \leq \log n - \log\log n$
\end{corollary}

% It is easily verifiable that Delsarte's dual is obtained
% when $\ell=1$, either with or without \eqref{eq:dual_ell6}.

\clearpage
\subsection{Linear Codes - $\DelsarteLin{n}{d}{\ell}$}

% For linear codes we are presently able to adapt
% the function $\Phi$, from step \ref{step1}.
% carry out only
% the first step of the framework. Namely, we
% find a low degree polynomial 
% $\Phi^{\Lin}:\cube{\ell\times n}\to\R$ which 
% satisfies  
% $\Phi^{\Lin}(X) \leq 0$ if \mbox{$0\neq X\notin \ForbLin{n}{d}{\ell}$} and
% $\Phi^{\Lin}(0) > 0$. 
% Our construction is very similar to the
% construction for general codes, for which we were able to find
% an analytic solution.
% , and it results in a polynomial in $n$ of degree
% $\ell \cdot 2^\ell$.

Let us define the dual of $\DelsarteLin{n}{d}{\ell}$.
\begin{proposition}
    \label{prop:dual_ell_linear}
    $A_{\Lin}(n,d)^\ell$ is upper bounded by
    \begin{align}
        & \textup{minimize} && g(0)/\widehat{g}(0)
        % \tag{\ref{eq:dual_ell1}}
        \nonumber
        \\
        & \textup{subject to} && g:\cube{\ell \times n} \to \R
        % \tag{\ref{eq:dual_ell2}}
        \nonumber
        \\
        &&& \widehat{g} \geq 0 
        % \tag{\ref{eq:dual_ell3}}
        \nonumber
        \\
        &&& \widehat{g}(0) > 0
        % \tag{\ref{eq:dual_ell4}}
        \nonumber
        \\
        &&& g(X) \leq 0 && X\in \ValidLin{n}{d}{\ell}\wedge X\neq 0 % \setminus\{0\}
        \tag{\ref{eq:dual_ell5}b}
        \label{eq:dual_ell6}
    \end{align}
    where $\ValidLin{n}{d}{\ell}$ is the complement of the set
    of forbidden configurations for linear codes,
    \[
        \Valid{n}{d}{\ell}
        \coloneqq 
        \{
            (x_1,\dots,x_\ell)\in\left(\cube{n}\right)^\ell:
            y = 0 \vee |y|\geq d
            \textup{ for all } y \in span(x_1,\dots,x_\ell)
        \}
    \]
\end{proposition}
Note that the only difference between the dual
of $\Delsarte{n}{d}{\ell}$ and that of $\DelsarteLin{n}{d}{\ell}$
is that constraint \eqref{eq:dual_ell5} is replaced
by \eqref{eq:dual_ell6}.

Based on the function $\Phi$ from the previous section
we create a function $\Phi^{\Lin}$ which is non-positive
on $\ValidLin{n}{d}{\ell}$.
The basic building blocks of $\Phi(X)$ were the functions
$\{n-2|x_i|\}$, where $x_1,\dots,x_\ell$ are the rows of
$X$. Namely, $\Phi(X)$ acts separately
and symmetrically on each row of $X$. Therefore 
a solution for $\ell=1$
can be transformed to a solution for larger $\ell$,
% by the tensor product $\Lambda^{\otimes \ell}$,
as done in the previous section.

For linear codes, however, we need to consider
linear combinations of $X$'s rows.
Thus, $\Phi^{\Lin}$ is built from the functions
% we need to use a more diverse set 
% of functions, 
$\{n-2|u^\top X|\}$ for $0\neq u\in\cube{\ell}$.
This is the set of linear {\em multivariate
Krawtchouk polynomials}, a family of multivariate orthogonal
polynomials. The classical Krawtchouk polynomials play
a key role in earlier studies of the rate vs.\ distance problem.
The multivariate Krawtchouk polynomials occupy an analogous position
in the present theory.
For more on these polynomials and their relation
to the LP hierarchies, see \cite{coregliano2021complete,loyfer2022new}.
The coefficient matrix of $n-2|u^\top X|$ in Fourier basis
is the $2^{\ell n}\times 2^{\ell n}$ matrix which
we denote by $A^u$, for any non-zero $u\in\cube{\ell}$.
This matrix is defined, for every $X,Y\in\cube{\ell\times n}$,
by
\begin{equation}
    \label{eq:Au_def}
    A^{u}_{X,Y} = 
    \begin{cases}
        1 & \text{if } X + Y = ue_j^\top
        \text{ for some } j=1,\dots,n
        \\
        0 & \text{otherwise}
    \end{cases}
\end{equation}
where $e_j$ is the $j$-th standard basis vector in $\cube{n}$.
The proof is a one-liner similar to \eqref{eq:proof_K1_fourier}.
Notice that when $\ell=1$ this is the adjacency matrix
of the Hamming cube $\cube{n}$.

We turn to define $\Phi^{\Lin}$.

Let $m\in\N$ be even such that $2^{\ell-1} \leq \left(\frac{1+\delta}{1-\delta}\right)^m$.

Define $\phi^\Lin_v:\cube{\ell\times n}\to\R$
\[
    \phi^\Lin_v(X) = 
    \sum_{u:\langle u,v\rangle_{\mathbb{F}_2}=1}
    \left[(n+d-2|u^\top X|)^m - (n-d)^m\right]
    % \frac{1-\chi_v(u)}{2}\varphi^m(u^T X)
\]
for $0\neq v \in \cube{\ell}$. Define $\Phi^{\Lin}=\Phi^{\Lin}_{n,d,\ell}$ by
\[
    \Phi^{\Lin}(X) = \prod_{0\neq v\in\cube{\ell}} \phi^\Lin_v(X)
\]
Here is the analogue of proposition \ref{prop:nonpos_poly_general} for 
$\Phi^\Lin$.
\begin{proposition}\label{prop:poly_linear}
    \hfill
    \begin{enumerate}
        \item $\phi^\Lin_v(X) \geq 0$ if $|u^\top X|=0$ for some $u$ for which $\langle v, u\rangle_{\mathbb{F}_2}=1$.
        \item $\phi^\Lin_v(X) \leq 0$ if $|u^\top X|\geq d$ for all $u$ for which $\langle v, u\rangle_{\mathbb{F}_2}=1$.
        \item $\Phi^{\Lin}(0)>0$.
        \item $\Phi^{\Lin}(X) \leq 0$ if $X\in \ValidLin{n}{d}{\ell}$ and $X \neq 0$.
    \end{enumerate}
\end{proposition}

One way to proceed to a feasible solution is by
solving the 
following problem, which is based on the ideas
from proposition \ref{prop:lambda_sufficient_conditions}.
% Using $\Phi^{\Lin}$, a feasible dual solution to
% $\DelsarteLin{n}{d}{\ell}$ can be found by solving the following
% problem:
\begin{prob}
    \label{prob:dual_fesible_linear}
    % Find $\Gamma:\cube{\ell\times n}\to\R$ such that
    \begin{align*}
        &\underset{\Gamma:\cube{\ell\times n}\to\R}
            {\textup{minimize}}
        \qquad
        \left| supp(\widehat{\Gamma}) \right|
        % \qquad
        \\
        &\textup{subject to}\qquad
        \widehat{\Gamma}(0) = 1;\quad
        \widehat{\Gamma}\geq 0;\quad
        \widehat{\Phi}_{\Lin}* \widehat{\Gamma} \geq 
        2^{(\ell-1)(2^\ell-1)} \widehat{\Gamma}
    \end{align*}
    % and the support size of $\widehat{\Gamma}$ is minimal.
\end{prob}
Solving problem \ref{prob:dual_fesible_linear} would yield the
following bound.
\begin{proposition}
    \label{prop:gamma_feasible_linear_ell}
    Let $\Gamma$ be a solution to problem \ref{prob:dual_fesible_linear}. Then
    \[
        A_{\Lin}(n,d) \leq \left(e n^{1/\delta}\right)^{2^{\ell}}
        \left|supp(\widehat{\Gamma})\right|^{1/\ell}
    \]
\end{proposition}
% Using the same arguments from the previous subsection,
% it is immediate that $g = \Phi^{\Lin}\cdot \Gamma^2$
% is feasible, and it yields the bound

Let us comment on the
tensor product
$\Lambda^{\otimes \ell}$ from the previous
section, and why it is not a viable choice here.
In the proof of proposition \ref{prop:dual_feasibl_sol_ell},
we rely on the fact that
\begin{equation}
    \label{eq:A_lambda_ell}
    \mathcal{F}[K_1(n-|x_i|)]
    * \widehat{\Lambda}^{\otimes \ell}
    = A^{e_i} \widehat{\Lambda}^{\otimes \ell}
    \geq (n-2(d-\ve)) \widehat{\Lambda}^{\otimes \ell}
\end{equation}
where $e_i$ is the $i$-th standard basis
vector in $\cube{\ell}$.
An analogous proof that
$\Lambda^{\otimes \ell}$
is feasible for problem \ref{prob:dual_fesible_linear}
% For a similar proof to work here, 
requires 
\eqref{eq:A_lambda_ell} to apply to all $0\neq u\in\cube{\ell}$, namely
\[
    A^{u} \widehat{\Lambda}^{\otimes \ell}
    \geq (n-2(d-\ve)) \widehat{\Lambda}^{\otimes \ell}
\]
But this is not the case.
Indeed, the definition
of $\Lambda$ implies 
$\widehat{\Lambda}(x) \leq (n-2d)^{-1}$ for every $|x|=1$. 
Let 
$u\in\cube{\ell}$ with $|u|\geq 2$, then
\[
    (A^u \widehat{\Lambda}^{\otimes \ell})(0)
    =\widehat{\Lambda}^{\ell-|u|}(0) \sum_{i=1}^{n} \widehat{\Lambda}^{|u|}(e_i)
    \leq \frac{n}{(n-2d)^{|u|}}
    < (n-2(d-\ve)) \widehat{\Lambda}^{\otimes \ell}(0)
\]

\section{Dual Feasible Solution to the Linear-Valued Objective}

\label{section:dual_sol_lin}

Changing the objective function of $\DelsarteLin{n}{d}{\ell}$
yields a very different dual problem. We recall the new objective,
which bounds $A_{\Lin}(n,d)$ instead of $A_{\Lin}(n,d)^\ell$:
\begin{equation}
    \label{eq:objective_linear}
    \textup{maximize} \quad 
    \frac{1}{2^{\ell}-1} \sum_{0\neq u \in\cube{\ell}}
    \sum_{x\in\cube{n}} f(ux^\top)
\end{equation}
where $f:\cube{\ell\times n}\to\R$ is feasible for
$\DelsarteLin{n}{d}{\ell}$.
A particular advantage of this objective function
is that now the LP is well-defined 
when $\ell\to\infty$. 
We believe that there is much to be gained from this fact.
Another advantage is this: 
Whereas our construction from the previous step
becomes too weak when $\ell$ is too large,
the solutions that we provide here
are good for any $\ell$.

% that this LP is well-defined when $\ell\to\infty$. 
% % We leave this to future work.

% A particular advantage of this modification is 
% that it simplifies the construction of a dual feasible solution,
% as we shall see shortly.
% Unlike our construction from the previous section, the 
% solution we provide here attains the first LP bound
% {\em for any} $\ell$. 

Recall the completeness theorem of \cite{coregliano2021complete},
which states, informally,
that $\DelsarteLin{n}{d}{\ell}$ converges to the true
value of $A_{\Lin}(n,d)$ when $\ell = \Omega(n^2)$. 
The proof of this theorem does not apply
when the objective function is \eqref{eq:objective_linear},
however
% We note, however, this does
% not apply to the objective function we discuss in this
% section. We are encouraged by 
numerical results from \cite{loyfer2022new}
% which 
show that, at least for $\ell=2$, 
the objective function \eqref{eq:objective_linear}
is on par
with the objective function of $\DelsarteLin{n}{d}{\ell}$.

Let us define the dual problem.
% Standard arguments yield
% the dual LP, as follows.
\begin{proposition}
    \label{prop:dual_linear}
    $A_{\Lin}(n,d)$ is upper bounded by
    \begin{align}
        & \textup{minimize} && g(0)
        \label{eq:dual_lin1}
        \\
        & \textup{subject to} && g:\cube{ \ell \times n}\to\R
        \label{eq:dual_lin2}
        \\
        &&& \widehat{g} \geq 0
        \label{eq:dual_lin3}
        \\
        &&& \widehat{g}(0) = 1
        \label{eq:dual_lin4}
        \\
        &&& g(X) \leq 1 && X \in \ValidLin{n}{d}{\ell}
        \label{eq:dual_lin5}
        \\
        &&& g(ux^\top) \leq 1-\frac{1}{2^\ell-1}
            && |x| \geq d\wedge u \neq 0
            \label{eq:dual_lin6}
    \end{align}
\end{proposition}
The last constraint states, in other words, that if 
$X \in \ValidLin{n}{d}{\ell}$ and its rank is $1$, then
$g(X)\leq 1 - 1/(2^\ell-1)$.

Let us proceed in finding a feasible solution.

\begin{proposition}
    \label{prop:dual_linear_sol}
    Let $g_1$ be {\em any} dual feasible solution to Delsarte's LP.
    Namely,
    \[
        \widehat{g}_1 \geq 0,~
        \widehat{g}_1(0) = 1,~
        g_1(x) \leq 0 ~\text{if}~ |x|\geq d
    \]
    Let
    \[
        g(X) = 
        \begin{cases}   
            1 + \frac{1}{2^\ell-1}(g_1(x) - 1)
                & X = ux^\top,~ 0\neq u \in \cube{\ell},~
                0\neq x\in \cube{n}
            \\
            g_1(0) & X = 0
            \\
            1 & \text{otherwise}
        \end{cases}
    \]
    Then, $g$ is feasible for the LP defined in proposition
    \ref{prop:dual_linear}, and its value is $g_1(0)$.
\end{proposition}
The value of $g$ is equal to that of $g_1$ by construction.

Constraint \eqref{eq:dual_lin5} is satisfied because $g(X)=1$
for {\em every} $X$ of rank $\geq 2$, regardless of the weights
of its span. Constraint \eqref{eq:dual_lin6} is satisfied because
$g_1(x) \leq 0$ when $|x|\geq d$.

For the remaining constraints we need the following proposition.
\begin{proposition}
    \label{prop:g_hat_linear}
    \[
        \widehat{g}(X)
        = \delta_{0}(X)
        + \frac{2^{-(\ell-1) n}}{2^\ell -1} 
        \sum_{0\neq u \in \cube{\ell}} 
            \widehat{g}_1(u^\top X) - \delta_0(u^\top X)
    \]
\end{proposition}
Constraints \eqref{eq:dual_lin3} and
\eqref{eq:dual_lin4}
follow from the proposition and the facts that $\widehat{g}_1(0)=1$ and
$\widehat{g}_1 \geq 0$.

\section{Discussion}
\label{section:discussion}

The new LP hierarchies
\cite{coregliano2021complete,loyfer2022new}
open a new way to engage with the rate vs.\
distance problem for linear codes.
In this work, we leverage proofs of
the first LP bound to develop the {\em first}
family of feasible solutions for these LPs, which
attain the bound.

For the $\Delsarte{n}{d}{\ell}$ hierarchy,
our solutions
recover
the first LP bound in the
range $\ell \leq \log n -\log\log n$. 
It is known that good 
solutions exist for all $\ell$, and we intend to return
in future work to the search of such solutions.
The holly grail of this research is proofs of tighter upper
bounds on $\mathcal{R}_\Lin(\delta)$. A possible approach
starts from the observation that
a solution for $\Delsarte{n}{d}{\ell}$ is also feasible
for linear codes. To this end we will seek modifications
of such solutions, as indicated above.

For $\DelsarteLin{n}{d}{\ell}$, the hierarchy for linear codes,
we introduced problem \ref{prob:dual_fesible_linear}.
It is based on the same methods we used for general codes.
Although we still do not know whether
good solutions for this problem will improve the bound,
we believe that
a better understanding of this problem,
and in particular of the operators $\{A^{u}\}_{0\neq u \in\cube{\ell}}$ (see \eqref{eq:Au_def}),
% finding a solution,
% even one that is equivalent to the first LP bound,
will resolve many of the remaining mysteries.

We also considered another objective function
for $\DelsarteLin{n}{d}{\ell}$, that bounds $A_{\Lin}(n,d)$
rather than $A_{\Lin}(n,d)^\ell$. This hierarchy has
the advantage that is it well 
defined when $\ell\to\infty$. The solutions we
construct for this problem match the first LP bound for
every $\ell$.

\bibliography{refs}
\bibliographystyle{IEEEtran}

\appendix

\section{Proofs}\label{section:proofs}

\begin{proof}[Proof of proposition \ref{prop:delsarte_dual}]
    This is a particular case of proposition
    \ref{prop:dual_ell_general}, with $\ell=1$.
\end{proof}

\begin{proof}[Proof of proposition \ref{prop:lambda_sufficient_conditions}]
    By construction, $g$ satisfies constraint
    \eqref{eq:delsarte_dual5}.
    For constraint \eqref{eq:delsarte_dual4}, 
    by the convolution theorem,
    \[
        \widehat{g}
        = \mathcal{F}[2(d-|x|)] * \widehat{\Lambda}
        * \widehat{\Lambda}
        = \left((A - (n-2d)I)\widehat{\Lambda}\right)
        * \widehat{\Lambda}
        \geq 2\ve \widehat{\Lambda}* \widehat{\Lambda}
        \geq 0
    \]
    where $A$ is the adjacency matrix of the Hamming cube, defined in \eqref{eq:A_def}.
    By the preceding equation,
    \[
        \widehat{g}(0) 
        \geq 2\ve (\widehat{\Lambda}* \widehat{\Lambda})(0)
        = 2\ve \lVert \widehat{\Lambda} \rVert_2^2
        > 2\ve \widehat{\Lambda}^2(0) > 0
    \]
    hence $g$ satisfies \eqref{eq:delsarte_dual3}. Finally, let us bound the value of $g$:
    \begin{equation}
        \label{eq:bound_lambda_support}
        \frac{g(0)}{\widehat{g}(0)}
        \leq \frac{2d \Lambda^2(0) }{2\ve (\widehat{\Lambda}* \widehat{\Lambda})(0)}
        = \frac{d}{\ve} \frac{\lVert\widehat{\Lambda} \rVert_1^2}{\lVert \widehat{\Lambda} \rVert_{2}^2}
    \end{equation}
    Note that $\lVert\widehat{\Lambda} \rVert_1
    = \langle \widehat{\Lambda}, \1_{supp(\widehat{\Lambda})}
    \rangle$, and apply Cauchy-Schwartz inequality to complete the proof.
\end{proof}

\begin{proof}[Proof of proposition \ref{prop:lambda_existence}]
    \hfill

    We will use the following facts. References can be found,
    e.g., in \cite{mceliece1977new}.
    \begin{fact}
    The roots of the Krawtchouks all lie in $(0,n)$.
    Denote by $z_{j,i}$ the $j$-th root of $K_i$.
    The roots of $K_i$ and $K_{i+1}$
    interlace: $z_{j,i+1}<z_{j,i}<z_{j+1,i+1}$
    for $i=1,\dots,n-1$ and $1\leq j \leq i-1$. 
    $K_i$ is strictly positive in $[0,z_{1,i})$.
    \end{fact}
    \begin{fact}
        \label{fact:krawtchouk_zero_asymptotic}
        For $n$ large enough,
        \[
            z_{1,k} = n/2 - \sqrt{k(n-k)} + o(n) 
        \]
    \end{fact}
    \begin{fact}[Chritoffel Darboux formula]
    Let $0\leq j \leq n$ and define
    \begin{equation}
        \Lambda_j(t,s) \coloneqq \sum_{i=0}^{j}\binom{n}{i}^{-1} K_i(t)K_i(s)
        \label{eq:CD_eigenfuncs}
    \end{equation}
    for every $t,s\in \R$.
    Then,
    \begin{equation}
        \left[K_1(s) - K_1(t)\right] \Lambda_j(t,s) = 
        \frac{j+1}{\binom{n}{j}}\left[ K_{j+1}(s)K_{j}(t) - K_{j}(s) K_{j+1}(t) \right]
        \label{eq:CD}
    \end{equation}
    \end{fact}
    % Let $\ve > 0$.
    Let us define $\Lambda$.
    
    Let $r\in \N$ be smallest such that 
    $z_{1,r} \geq d-\ve $,
    where $z_{1,r}$ is the first root of the $r$-th Krawtchouk, $K_r$.
    This implies $d-\ve \in [z_{1,r+1}, z_{1,r}]$.
    
    Define 
    \begin{equation}
        \Lambda(x)
        = \Lambda_{d,\ve}(x)\coloneqq \Lambda_{r}(x,d-\ve)
        = \sum_{i=0}^{r}\binom{n}{i}^{-1} K_{i}(d-\ve) K_{i}(x)
        \label{eq:Lambda}
    \end{equation}
    where $\Lambda_r(x,d-\ve)$ was defined in \eqref{eq:CD_eigenfuncs}.
    
    The Fourier transform of $\Lambda$ is
    \begin{equation}
        \label{eq:Lambda_hat}
        \widehat{\Lambda}(x)
        = \sum_{i=0}^{r} \binom{n}{i}^{-1} K_{i}(d-\ve) L_{i}(x)
    \end{equation}
    because $\widehat{K}_i = L_i$, which is the indicator of 
    the set $\{x\in\cube{n}: |x|=i\}$.
    
    Let us show that $\Lambda$ satisfies proposition
    \ref{prop:lambda_sufficient_conditions}.
    \begin{enumerate}[label=(\alph*)]
        \item By \eqref{eq:Lambda_hat}, $\widehat{\Lambda}(0) = 1$.
        \item Recall that $K_i$ is positive in the segment $[0,z_{1,i})$;
        that $z_{1,i} > z_{1,r}$ for all $i<r$;
        and we chose $r$ so that $z_{1,r}\geq d-\ve$, whence $K_{i}(d-\ve)  \geq 0$ for
        every $0\leq i \leq r$.
        
        Therefore, $\widehat{\Lambda} \geq 0$.
        
        \item The degree-1 Krawtchouk is $K_1(t)=n-2t$. We rearrange $2(d-|x|)$ by adding and subtracting $2\ve$ and
        writing it using $K_1$.
        \[
            2(d-|x|) = 2\ve + K_1(x) - K_1(d-\ve)
        \]
        Then apply
        the Christoffel-Darboux
        formula \eqref{eq:CD}:
        % hence
        \[
            \mathcal{F}(2(d-|x|) \cdot \Lambda)(x)
            = 2\ve \widehat{\Lambda}(x)
            +\frac{r+1}{\binom{n}{r}} \left[
            K_{r}(d-\ve) L_{r+1}(x)
            - K_{r+1}(d-\ve) L_{r}(x)
            \right]
        \]
        The first term is non-negative by the previous item, and
        the rest is also non-negative by our choice of $r$. Therefore,
        \[
            2\widehat{(d-|x|)}*\widehat{\Lambda}
            = A \widehat{\Lambda} - (n-2d) \widehat{\Lambda}
            \geq 2\ve \widehat{\Lambda}
        \]
    \end{enumerate}
    
    Finally, note that $\widehat{\Lambda}$ is supported on the
    Hamming ball of radius $r$.
    \begin{remark}
    Our proof here is based on Krawtchouk theory and is close
    to \cite{mceliece1977new}. It works just as well with
    \[
        \Lambda_r(x,z_{1,r+1}) = \sum_{i=0}^{r}\binom{n}{i}^{-1} 
            K_i(z_{1,r+1}) K_i(x)
    \]
    which is the $\Lambda$ we
    described at the end of section \ref{section:first_lp_bound}
    The first zero of $K_{r+1}$
    is the spectral radius of $A^{\leq r}$,
    the adjacency matrix of the Hamming ball of radius $r$.
    \end{remark}
    
\end{proof}

\begin{proof}[Proof of corollary \ref{cor:first_lp_bound}]
    The cardinality of the Hamming ball of radius $r$
    is $2^{nH(r/n)+o(1)}$. 
    Choose $\ve$ not too small in proposition \ref{prop:lambda_sufficient_conditions}, e.g.\ $\ve=1$.
    By propositions
    \ref{prop:delsarte_dual}, \ref{prop:lambda_sufficient_conditions} and \ref{prop:lambda_existence},
    \[
        A(n,d) \leq 2^{nH(1/2-\sqrt{d/n(1-d/n)})+o(n)}
    \]
\end{proof}

\begin{proof}[Proof of proposition \ref{prop:dual_ell_general}]
    % The dual can be obtained by standard LP methods. Here we 
    % provide a direct proof of the proposition.
    
    Let $g$ be a feasible solution to the LP in the proposition.
    Let $f$ be a feasible solution to $\Delsarte{n}{d}{\ell}$.
    
    % Explanation follows.
    \begin{align*}
        \widehat{g}(0) \sum_{X\in\cube{\ell\times n}} f(X)
        &= 2^{\ell n} \widehat{g}(0)\widehat{f}(0)
        \\
        &\leq 2^{\ell n} \sum_{X} \widehat{g}(X)\widehat{f}(X)
        \\
        &= 2^{\ell n} \langle \widehat{g}, \widehat{f} \rangle_{\mathcal{F}}
        \\
        &= 2^{\ell n} \langle g, f \rangle
        \\
        &= \sum_{X} g(X)f(X)
        \\
        &\leq g(0)
    \end{align*}
    The first transition if by definition. The second is 
    because $\widehat{f}\geq 0$ and $\widehat{g}\geq 0$.
    The fourth is by Parseval's identity.
    The last transition is because, for each $0\neq X\in\cube{\ell\times n}$,
    if $X\in \Forb{n}{d}{\ell}$ then $f(X) = 0$,
    otherwise $g(X)\leq 0$ and $f(X) \geq 0$.
    
    Finally, we use the fact that $A(n,d)^\ell \leq \sum_X f(X)$.
\end{proof}

\begin{proof}[Proof of proposition \ref{prop:nonpos_poly_general}]
    \hfill
    \begin{enumerate}
        \item $(n+d-2|x_i|)^m \geq 0$ for every $i$ because $m$ is even.
        If $x_i = 0$ for some $i\in U$ then
        \[
            \sum_{i\in U} (n+d+|x_i|)^m \geq (n+d)^m > \ell(n-d)^m \geq |U|(n-d)^m
        \]
        The second inequality follows from the constraint on $m$.
        \item Always $|x|\leq n$, so if $|x|\geq d$
        % \[
        %     -n+d \leq n+d-2|x| \leq n-d
        % \]
        \[
            \big| n+d-2|x| \big| \leq n-d
        \]
        since $m$ is even, $(n+d-2|x_i|)^m \leq (n-d)^m$. Assuming $|x_i| \geq d$
        for all $i\in U$,
        \[
            \phi_U(x_1,\dots,x_r) \leq |U|\left[(n-d)^m -(n-d)^m\right] = 0
        \]
        \item Obvious.
        \item Let $0\neq X\notin \Valid{n}{d}{\ell}$.
        Let $V = \{1\leq i \leq \ell: |x_i|\geq d\}$. Then $x_i = 0$ if $i\notin V$.
        By item 1, $\phi_U(X) > 0$ for every $U\not\subset V$. By item 2,
        $\phi_U(X)\leq 0$ for every $U\subset V$. There are $2^{|V|}-1$ non-empty
        subsets of $V$. $\Phi(X)$ is a product of an odd number of non-positive functions,
        and some positive functions. Hence $\Phi(X) \leq 0$.
    \end{enumerate}
\end{proof}

\begin{proof}[Proof of proposition \ref{prop:dual_feasibl_sol_ell}]
\hfill
\begin{enumerate}
    \item By proposition \ref{prop:nonpos_poly_general},
    $g$ satisfies \eqref{eq:dual_ell6}.
    
    It remains to show that $\widehat{g} \geq 0$ and $\widehat{g}(0)>0$.
    
    % Since we are dealing with powers it will be more convenient
    % to use matrix notation. Recall that $\widehat{K}_1 = L_1$,
    % where $L_1$ is the indicator of the set $\{x\in\cube{n}: |x|=1\}$.
    % Consider the operator of convolution with $L_1$.
    % The matrix of this operator 
    % is the  $2^n\times 2^n$ matrix $A$, the adjacency matrix of Hamming cube.
    % Namely, for any $x, y\in \cube{n}$
    % \[
    %     A_{x,y} =
    %     \begin{cases}
    %         1 & |x+y| = 1 \\
    %         0 & \text{otherwise}
    %     \end{cases}
    % \]
    % We include the simple proof: let $f:\cube{n}\to\R$,
    % \[
    %     2^n(L_1*f)(x) = \sum_{y\in\cube{n}} L_1(y) f(x+y)
    %     = \sum_{i=1}^{n} f(x+e_i)
    %     = \sum_{y: |y+x|=1} f(y)
    %     = (Af)(x)
    % \]
    In the previous section we saw that
    % $\widehat{\phi} * \widehat{\Lambda} 
    % \geq 2\ve \widehat{\Lambda}$, where $\phi(x)=K_1(x)-K_1(d)$. In
    % matrix notation we have
    \[
        (A - (n-2d)I) \widehat{\Lambda}
        \geq 2\ve \widehat{\Lambda}
    \]
    which implies
    \[
        (A + dI) \widehat{\Lambda}
        \geq (n  - d + 2\ve) \widehat{\Lambda}
    \]
    Repeated application of the operator $A + dI$
    results in
    \[
        (A + dI)^m \widehat{\Lambda}
        \geq (n  - d + 2\ve)^m \widehat{\Lambda}
    \]
    Let $i \in [\ell]$. 
    The function $(n+d-2|x_i|)^m$ can be expressed as
    \[
        (K_1(x_i) + d)^m \cdot \prod_{j\in[\ell], j\neq i} K_0(x_j)
    \]
    because $K_0\equiv 1$.
    The Fourier transform of $K_0$ is $L_0$, and convolution with $L_0$ corresponds to the identity matrix $I$.
    Thus, convolution with $\mathcal{F}[(n+d-2|x_i|)^m]$ corresponds to
    the matrix
    \[
        I \otimes \dots \otimes I \otimes \underbrace{(A + d)^m}_{i\text{-th coordinate}}
        \otimes I \otimes \dots \otimes I
    \]
    Namely, the
    convolution operator of
    $\mathcal{F}[(n+d-2|x_i|)^m]$ interacts only with the
    $i$-th coordinate in $\left(\cube{n}\right)^{\ell}$, hence
    \[
        \mathcal{F}[(n+d-2|x_i|)^m] * \widehat{\Lambda}^{\otimes \ell} \geq (n-d+2\ve)^m \widehat{\Lambda}^{\otimes \ell}
    \]
    By linearity of the convolution operation, and 
    by our choice of $\ve$,
    \[
        \widehat{\phi}_{U} * \widehat{\Lambda}^{\otimes \ell}
        \geq |U| \big(
        (n-d+2\ve)^m
        - (n-d)^m \big) \widehat{\Lambda}^{\otimes \ell}
        = |U| \widehat{\Lambda}^{\otimes \ell}
    \]
    for every $\emptyset\neq U \subset [\ell]$.
    Thus,
    \[
        \widehat{\Phi} * \widehat{\Lambda}^{\otimes \ell}
        \geq \left( \prod_{j=1}^{\ell} j^{\binom{\ell}{j}} \right)
        \widehat{\Lambda}^{\otimes \ell}
    \]
    This implies that $\widehat{g}\geq 0$ and $\widehat{g}(0) >0$,
    namely $g$ is feasible.
    
    \item Let us compute the value of $g$. Using similar reasoning
    as in \eqref{eq:bound_lambda_support},
    \[
        % A(n,d)^{\ell} 
        % \leq 
        \frac{g(0)}{\widehat{g}(0)}
        \leq \frac{\Phi(0)}{\prod_{j=1}^{\ell}j^{\binom{\ell}{j}}}
        \left| supp(\widehat{\Lambda}^{\otimes \ell}) \right|
        = \frac{\Phi(0)}{\prod_{j=1}^{\ell}j^{\binom{\ell}{j}}}
        \left| supp(\widehat{\Lambda}) \right|^{\ell}
    \]
    We can pick $m\geq \frac{1}{\delta} \log \ell$. Then,
    % but, choosing $m = (2\delta)^{-1} \log\ell $,
    \begin{align*}
        \Phi(0) 
        &\leq \prod_{j=1}^{\ell}
        \left(j n^m ((1+\delta)^m
        -(1-\delta)^{m})\right)^{\binom{\ell}{j}} 
        \\
        &\leq \left(\prod_{j=1}^{\ell}j^{\binom{\ell}{j}}\right)
        \left(
        n^m e^{\delta m}
        \right)^{\sum_{j=1}^{\ell}\binom{\ell}{j}}
        \\
        &\leq \left(\prod_{j=1}^{\ell}j^{\binom{\ell}{j}}\right)
        n^{\frac{2^{\ell}\log \ell}{\delta}} \ell^{2^{\ell-1}}
    \end{align*}
    hence
    \[
        \frac{g(0)}{\widehat{g}(0)} \leq 
        \left(
        en^{1/\delta}
        \right)^{2^{\ell} \log \ell}
        \left| supp(\widehat{\Lambda}) \right|^{\ell}
    \]
\end{enumerate}
\end{proof}

\begin{proof}[Proof of corollary \ref{cor:dual_ell}]
    By propositions \ref{prop:dual_ell_general}
    and \ref{prop:dual_feasibl_sol_ell},
    \[
        A(n,d)\leq 
        \left(
        en^{1/\delta}
        \right)^{\frac{2^{\ell} \log \ell}{\ell}}
        \left| supp(\widehat{\Lambda}) \right|
    \]
    
    The value of $|supp(\hat{\Lambda})|$ is equivalent to the 
    first LP bound, by corollary \ref{cor:first_lp_bound}.
    Therefore, the bound we obtained is as long as
    if 
    \[
        \frac{1}{n}\log_2\left( 
        \left(
        en^{1/\delta}
        \right)^{\frac{2^{\ell} \log \ell}{\ell}}
        \right) = o_n(1)
    \]
    which is true when $\ell \leq \log n - \log\log n$.
\end{proof}

\begin{proof}[Proof of proposition \ref{prop:dual_ell_linear}]
    The proof is similar to that of 
    proposition \ref{prop:dual_ell_general}.
\end{proof}

\begin{proof}[Proof of proposition \ref{prop:poly_linear}]
    \hfill
    \begin{enumerate}
        \item From the first item of proposition \ref{prop:nonpos_poly_general} and by the choice of $m$,
        \[
            \phi_v(X) \geq (n+d)^m - 2^{\ell-1}(n-d)^m >0
        \]
        \item From the second item of proposition \ref{prop:nonpos_poly_general}, and since
        for every $v$ the number of $u$ for which $\chi_v(u)=-1$ is $2^{\ell-1}$,
        \[
            \phi_v(X) \leq 2^{\ell-1}(n-d)^m - 2^{\ell-1}(n-d)^m =0
        \]
        \item Obvious.
        \item Let $X\in \ValidLin{n}{d}{\ell}$, $X\neq 0$.
        Let $V = \{u: |u^T X| = 0\}$. Observe that $V$ is a 
        linear subspace.
        Let $v \in V^{\perp}\setminus\{0\}$. 
        By item 2, $\phi_v(X) \leq 0$. On the other hand, 
        if $v \in \cube{\ell}\setminus V^{\perp}$, by item 1 
        $\phi_v(X) > 0$. So $\Phi^{\Lin}(X)$ is a product of 
        $2^{\dim V}-1$
        non-positive functions, and $2^\ell - 2^{\dim V}$ 
        positive functions,
        hence $\Phi^{\Lin}(X) \leq 0$.
    \end{enumerate}
\end{proof}

\begin{proof}[Proof of proposition \ref{prop:gamma_feasible_linear_ell}]
    Using a similar reasoning to the proof of 
    propositions \ref{prop:lambda_sufficient_conditions} and \ref{prop:dual_feasibl_sol_ell}, it is not hard to see
    that $g\coloneqq \Phi^\Lin \Lambda^2$ is a feasible solution
    to $\DelsarteLin{n}{d}{\ell}$, with value
    \[
        \frac{g(0)}{\widehat{g}(0)}
        \leq 
        \frac{\Phi^\Lin(0)}{2^{(\ell-1)(2^{\ell}-1)}}
        \left| supp(\widehat{\Lambda})\right|
    \]
    Also, choosing $m \geq \ell/\delta$,
    \begin{align*}
        \Phi^\Lin(0)
        &\leq \left[ 2^{\ell-1} \left(
        (n+d)^m-(n-d)^m
        \right)\right]^{2^{\ell}-1}
        \\
        &\leq 2^{(\ell-1)2^{\ell-1}}
        (e^\delta n)^{m 2^{\ell}}
        \\
        & \leq 2^{(\ell-1)(2^{\ell}-1)}
        \left(en^{1/\delta}\right)^{\ell 2^{\ell}}
    \end{align*}
    Finally, recall that 
    $A_\Lin(n,d) \leq \left(g(0)/\widehat{g}(0)
    \right)^{1/\ell}$.
\end{proof}

\begin{proof}[Proof of proposition \ref{prop:dual_linear}]
    Let $f:\cube{\ell\times n} \to\R$ be a feasible 
    solution to $\DelsarteLin{n}{d}{\ell}$.
    Let $g:\cube{\ell\times n} \to\R$ be a feasible 
    solution to the prgoram in the proposition.
    \begin{align*}
        A(n,d)
        &\leq \frac{1}{2^\ell-1} 
        \sum_{0\neq u \in \cube{\ell}} 
        \sum_{x\in\cube{n}} f(ux^\top)
        \\
        &\overset{(1)}{=} 1 + \frac{1}{2^\ell-1} 
        \sum_{0\neq u \in \cube{\ell}} 
        \sum_{0\neq x\in\cube{n}} f(ux^\top)
        \\
        &\overset{(2)}{\leq} 1 + 
        \sum_{0\neq u \in \cube{\ell}} 
        \sum_{0\neq x\in\cube{n}} f(ux^\top)(1 - g(ux^\top))
        \\
        &\overset{(3)}{\leq} 1 + 
        \sum_{0 \neq X\in\cube{\ell\times n}} 
        f(X)(1 - g(X))
        \\
        &\overset{(1)}{\leq} 1 + 
        \sum_{X\in\cube{\ell\times n}} 
        (f(X) -\delta_0(X))(1 - g(X))
        \\
        &\overset{(4)}{=} 1 + 2^{\ell n} \langle \widehat{f} - \widehat{\delta},
        \widehat{\chi}_0 - \widehat{g}\rangle_{\mathcal{F}}
        \\
        &= 1 + 2^{\ell n} \langle \widehat{f} 
        - 2^{-rn}\chi_0,
        \delta_0 - \widehat{g}\rangle_{\mathcal{F}}
        \\
        &= 1 + g(0) -\widehat{g}(0)
        -  2^{\ell n} \langle \widehat{f},
        \widehat{g} - \delta_0 \rangle_{\mathcal{F}}
        \\
        &\overset{(5)}{\leq} g(0)
    \end{align*}
    \begin{enumerate}[label=(\arabic*)]
        \item $f(0) = 1$.
        \item For $u,x\neq 0$, if $|x|\leq d-1$ then $f(ux^\top) = 0$,
        otherwise $\frac{1}{2^\ell-1} \leq 1 - g(ux^\top)$.
        \item For $X\neq 0$ with rank $\geq 2$, if $X \in \ForbLin{n}{d}{\ell}$ then $f(X)=0$, otherwise $f(X) \geq 0$ 
        and $1-g(X)\geq 0$.
        \item Parseval's identity.
        \item $\widehat{g}(0)=1$, $\widehat{f}\geq 0$, 
        $\widehat{g} - \delta_0 \geq 0$.
    \end{enumerate}
\end{proof}

\begin{proof}[Proof of proposition \ref{prop:g_hat_linear}]
Rewrite $g$ in a more convenient way:
\[
    g(X) = 1
    + (g_1(0)-1)\delta_{0}(X)
    + \frac{1}{2^\ell-1}(g_1(x)-1) \1_{[X = ux^\top]}
\]
where $\1_{[X = ux^\top]}$ is the indicator function of the set
\[
    \{X\in\cube{\ell\times n}: X = ux^\top
    \text{ for some } 0\neq u\in\cube{\ell}, 0\neq x\in \cube{n} \}
\]
The constant function $1$ is the Fourier character that corresponds
to the zero vector, $\chi_{0}$. Its 
Fourier transform is Kronecker's delta function at $0$, $\delta_{0}(X)$.
\begin{align*}
    \widehat{g}(X)
    &= \delta_{0}(X)
        + 2^{-\ell n} \left[ (g_1(0)-1)\delta_{0}(X)
        + \frac{1}{2^\ell -1} \sum_{u\neq 0} 
        \sum_{y\neq 0} \chi_{X}(uy^\top) 
        \left[ g_1(y)-1\right]\right]
    \\
    &= \delta_{0}(X)
        + \frac{2^{-\ell n} }{2^\ell -1}
            \sum_{u\neq 0} \sum_{y} \chi_{X}(uy^\top)
            \left[ g_1(y)-1\right]
\end{align*}
It is not hard to verify that $\langle X, uy^\top \rangle = 
\langle u^\top X, y \rangle$, hence $\chi_X(uy^\top)=\chi_{u^\top X}(y)$.
Thus, the inner sum over $y$ is the projection of the functions
$g_1$ and $1$ over the Fourier character $\chi_{u^\top X}$, up to
normalization by $2^{-n}$.
 
\end{proof}

\end{document}